\documentclass[12pt]{article}
\pdfoutput=1
\usepackage{times}
\usepackage{natbib}
 \bibpunct{(}{)}{;}{a}{,}{,}
\usepackage{graphicx}
\usepackage{amsmath,amssymb,amsthm}
\usepackage{bm}
\usepackage{rotating}
\usepackage[vlined]{algorithm2e}
\usepackage{rotating}
\usepackage{multicol}
\usepackage{titlesec}
\usepackage{latexsym}
\usepackage[pdftex,colorlinks=true,linkcolor=blue,citecolor=blue,urlcolor=blue,bookmarks=false,pdfpagemode=None]{hyperref}
\usepackage{url}
\usepackage{verbatim}
\usepackage{fancyhdr}
\usepackage{setspace}
\usepackage{paralist}
\usepackage{boxedminipage}
\usepackage{color}
\definecolor{lightgrey}{rgb}{0.85,0.85,0.85}
\usepackage{lineno}
\usepackage[top=1in, bottom=1in, left=1in, right=1in]{geometry}

\newcommand{\bv}{\begin{array}}

\def\abovestrut#1{\rule[0in]{0in}{#1}\ignorespaces}
\def\belowstrut#1{\rule[-#1]{0in}{#1}\ignorespaces}
\def\abovespace{\abovestrut{0.20in}}
\def\aroundspace{\abovestrut{0.20in}\belowstrut{0.10in}}
\def\belowspace{\belowstrut{0.10in}}

\newtheorem{thm}{Theorem}
\newtheorem{deff}{Definition}

\newtheorem{lem}{Lemma}

\begin{document}

\title{Graphlet decomposition of a weighted network\protect\thanks{This paper will appear in {\it Journal of Machine Learning Research, Workshop \& Conference Proceedings}, vol. 22 (AISTATS), 2012. Address correspondence to EM Airoldi, \href{mailto:airoldi@fas.harvard.edu}{airoldi@fas.harvard.edu}.}}

\author{Hossein Azari Soufiani,~ Edoardo M. Airoldi\\
Department of Statistics\\
Harvard University, Cambridge, MA 02138, USA}
\date{}

\maketitle


\begin{abstract}
We introduce the {\it graphlet decomposition} of a weighted network, which encodes a notion of social information based on social structure.  
We develop a scalable inference algorithm, which combines EM with Bron-Kerbosch in a novel fashion, for estimating the parameters of the  model underlying graphlets using one network sample.
We explore some theoretical properties of the graphlet decomposition, including computational complexity, redundancy and expected accuracy. 
We demonstrate  graphlets on synthetic and real data. We analyze messaging patterns on Facebook and criminal  associations in the 19th century.

\end{abstract}

\newpage
\singlespacing 
\tableofcontents
\onehalfspacing
\newpage


\section{Introduction}

In recent years, there has been a surge of interest in 
 social media platforms such as Facebook and Twitter,
 collaborative projects in the spirit of Wikipedia,
 and services that rely on the social structure underlying these services, including 
  \href{www.cnn.com}{cnn.com}'s ``popular on Facebook'' and 
  \href{www.nytimes.com}{nytimes.com}'s ``most emailed''.
As these platforms and services have been gaining momentum,
efforts in the computational social sciences have begun studying patterns of behavior that result in organized social structure and interactions \citep[e.g., see][]{laze:pent:adam:aral:2009}. 
Here, we develop a new tool to analyze data about social structure and interactions routinely collected  in this context.

There is a rich literature of statistical models to analyze binary interactions or networks \citep{gold:zhen:fien:airo:2010}.
However, while many interesting interaction data sets involve weighted measurements on pairs on individuals, arguably, only a few of the existing models are amenable to analyze the resulting weighted networks.
We consider situations where we observe an undirected weighted network, encoded by a symmetric adjacency matrix with integer entries and diagonal elements equal to zero.
Our modeling approach is to decompose the graphon $\Lambda$, which defines an exchangeable model \citep{kall:2005}  for integer-valued measurements of pairs of individuals $P(Y|\Lambda)$, in terms of a number of basis matrices that grows with the size of the network,
\[
\textstyle
\Lambda = \sum_i \, \mu_i \, P_i.
\]
The factorization of $\Lambda$ is related to models for binary networks based on the singular value decomposition and other factorizations \citep{hoff:2009,kim:lesk:2010}.
However, 
 we abandon the popular orthogonality constraint \citep{jian:yao:liu:guib:2011} among the basis matrices $P_i$s to attain interpretability in terms of multi-scale social structure, and 
 we chose not model zero edge weights to enable the estimation to scale linearly with the number of positive weights \citep{lesk:chak:klei:falo:2010}.
These choices often lead to a non-trivial inferential setting \citep{airohaas:2010}.
Related methods utilize a network to encode inferred dependence among multivariate data \citep{coif:magg:2006,lee:nadl:wass:2008}.
We term our method ``graphlet decomposition'', as it is reminiscent of a wavelet decomposition of a graph. 
An unrelated literature uses the term graphlets to denote network motifs \citep{przu:corn:juri:2006,kond:sher:borg:2009}.

The key features of the graphlet decomposition we introduce in this paper are:
 the basis matrices $P_i$s are non-orthogonal, but capture overlapping communities at multiple scales;
 the information used for inference comes exclusively from positive weights;
 parameter estimation is linear in the number of positive weights, and 
 avoids dealing with permutations of the input adjacency matrix $Y$ by inferring the basis elements $P_i$ from data.

The basic idea is to posit a
 Poisson model for the edge weights $P(Y|\Lambda)$ and
 parametrize the rate matrix $\Lambda$ in terms of a binary factor matrix $B$.
 The binary factors can be interpreted as latent features that induce social structure through homophily.
 The factor matrix $B$ allows for an equivalent interpretation as basis matrices, which define overlapping cliques of different sizes.

Inference is carried out in two stages:
 first we identify candidate basis matrices using Bron-Kerbosch,
 then we estimate coefficients using EM.
 The computational complexity of the inference is addressed in Section \ref{sec:complexity}.


\section{Graphlet decomposition}
\label{secmod}

Consider observing an undirected weighted network, encoded by a symmetric adjacency matrix with integer entries and diagonal elements equal to zero.
In the derivations below, we avoid notation whose sole purpose is to set to zero  diagonal elements of the resulting matrices.

\subsection{Statistical model}

Intuitively, we want to posit a data generating process that can explain edge weights in terms of social information,  quantified by  community structure at multiple scales and possibly overlapping.
We choose to represent communities in terms of their constituent maximal cliques.

\begin{deff}\label{def1}
The graphlet decomposition (GD) of a matrix $\Lambda$ with non-negative entries $\lambda_{ij}$ is defined as
 $\Lambda= B\,W\,B'$,
where 
 $B$ is an  $N \times K$ binary matrix,
 $W$ is a $K \times K$ diagonal matrix.
Explicitly, we have
\begin{eqnarray*}
\left( \begin{matrix}
\lambda_{11}& \lambda_{12}&\ldots&\lambda_{1N}\\
\lambda_{21}& \lambda_{22}&\ldots&\lambda_{2N}\\
\vdots&\vdots&\ddots &\vdots\\
\lambda_{N1}& \lambda_{N2}&\ldots&\lambda_{NN}\end{matrix}\right)=
B\left( \begin{matrix}
\mu_{1}& 0&\ldots&0\\
0& \mu_{2}&\ldots&0\\
\vdots&\vdots&\ddots &\vdots\\
0& 0&\ldots&\mu_{K}\end{matrix}\right)B'
\end{eqnarray*}
where $\lambda_{ii}$ is set to zero and $\mu_i$ is positive for each $i$.
\end{deff}

The basis matrix $B$ can be interpreted in terms  overlapping communities at multiple scales.
To see this, denote the $i$-th column of the matrix $B$ as $b_{\cdot i}$.
We can re-write $\Lambda=\sum_{i=1}^K \mu_i \, P_i$, where each $P_i$ is an $N \times N$ matrix defined as $P_i = b_{\cdot i} \, b_{\cdot i}'$, in which the diagonal is set to zero.
We refer to $P_i$ as basis elements.
Think of the $i$-th column of $B$ as encoding a community of interest; then the vector $b_{\cdot i}$ specifies which individuals are members of that community, and the basis element $P_i$ is the binary graph that specifies the connections among the member of that community. The same individual may participate to multiple communities, and  the communities may span different scales. 

The model for an observed network $Y$ is then
\begin{eqnarray} 
 \label{model}
 \textstyle
 Y \sim \hbox{Poisson}_{+} \, ( \sum_{i=1}^K \mu_i \, P_i ).
\end{eqnarray}

This is essentially a binary factor model with a truncated Poisson link and two important nuances. 
Standard factor models  assume that entire rows of the matrix $Y$ are conditionally independent, even when the matrix is square and row and column $i$ refer to the same unit of analysis. 
In contrast, our model treats the individual random variables $y_{ij}$ as exchangeable, 
and imposes the restriction that positives entries in the factors map (in some way) to maximal cliques in the observed network. 
In addition, the matrix $Y$ is sparse; our model implies that zero edge weights carry no information about maximal cliques, 
and thus are not relevant for estimation. 
These nuances---map between cliques and factors, and information from positives entries only---produce a new and interesting inferential setting that scales to large weighted networks.


\subsection{Inference}

Computing the graphlet decomposition of a  network $Y$ involves the estimation of a few critical parameters: the number of basis elements $K$, the basis matrix $B$ or equivalently the basis elements $P_{1:K}$, and the coefficients associated with these basis elements $\mu_{1:K}$.

We develop a two-stage estimation algorithm.
First, we identify a candidate set of $K^c$ basis elements, using the Bron-Kerbosch algorithm, and obtain a candidate basis matrix $B^c$.
Second, we develop an expectation-maximization (EM) algorithm to estimate the corresponding coefficients $\mu_{1:K^c}$; several of these coefficients will vanish, thus selecting a final set of basis elements.
We study theoretical properties of this estimation algorithm in Section \ref{sec:theory}.
Figure \ref{figregraph} illustrates the steps of the estimation process on a toy network.

\begin{figure*}[t!]
 \centering
  \includegraphics[width=1.06\textwidth]{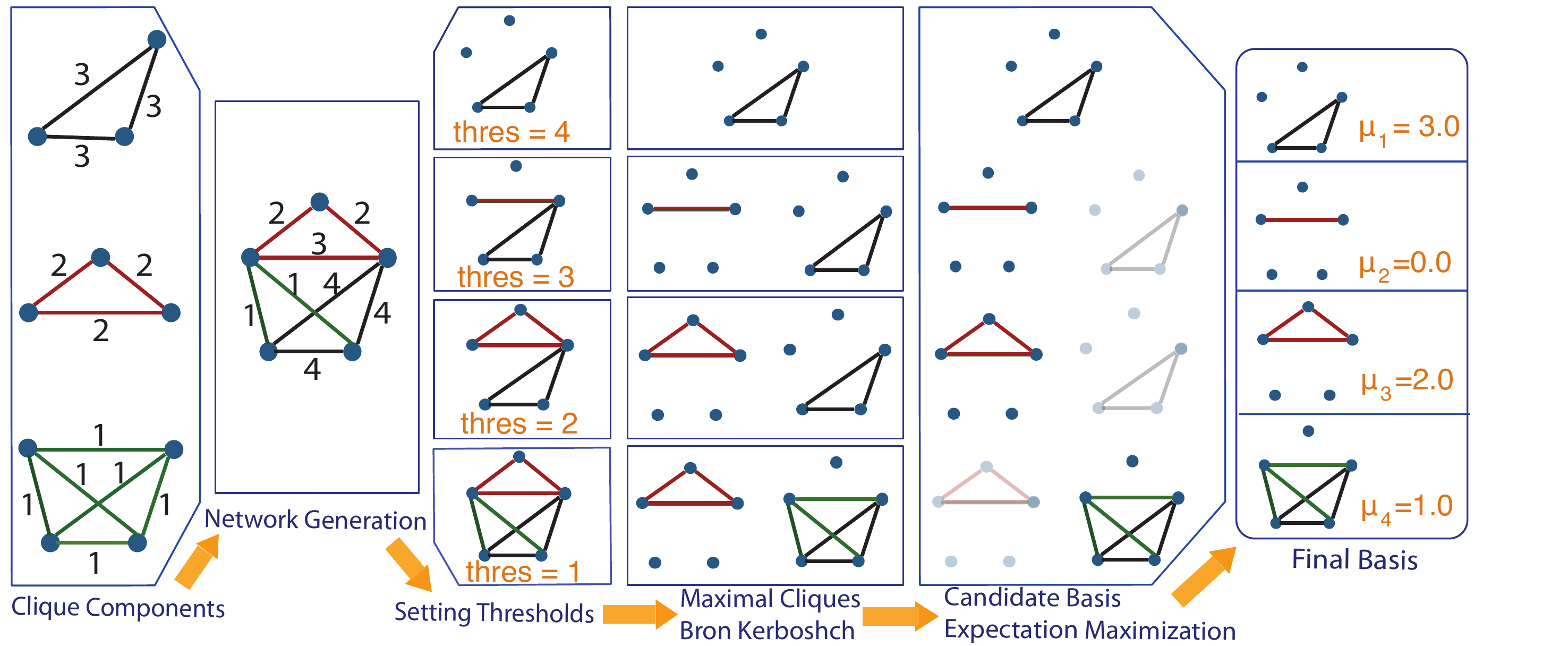}
  \caption{An illustration of the two-stage estimation algorithm for the graphlet decomposition.}\label{figregraph}
\end{figure*}


\subsubsection{Identifying candidate basis elements}

The algorithm to identify the candidate basis elements proceeds by thresholding the observed network $Y$ at a number of levels, $t=\max(Y)\dots\min(Y)$, and by identifying all the maximal cliques in the corresponding sequence of binary networks, $Y^{(t)}=\mathbf{1}(Y\ge t)$, using the Bron-Kerbosch algorithm \citep{BKcl}.
The resulting cumulative collection of $K^c$ maximal cliques found in the sequence of networks $Y^{(t)}$ is then turned into candidate basis elements, $P^c_i$, for $i=1\dots K^c$. Algorithm \ref{algo1} details this procedure.

\begin{algorithm}[ht!]
\label{algo1}
\textbf{\textit{}}\\
$B^c$ = empty basis set\\
\textbf{For}{ t = max$(Y)$ to min$(Y)$}\\
\ \ \ \ \ $Y^{(t)}=\textbf{1}(Y \ge t$)\\
\ \ \ \ \  C = maximal cliques($Y^{(t)}$) using Bron-Kerbosch\\
\ \ \ \ \ $B^c = B^c \cup C$\\
 \protect\smallskip
 \caption{Estimate a basis matrix, $B^c$, encoding candidate elements, $P^c$, from a weighted network, $Y$.}
\end{algorithm}

Algorithm \ref{algo1} allows us to avoid dealing with permutations of the input adjacency matrix $Y$ by inferring the basis elements $P_i$ independently of such permutation.

In Section \ref{sec:identifiability} we show that, under certain conditions on the true underlying basis matrix $B$, Algorithm  \ref{algo1} finds a set of candidate basis elements that contains the true set of basis elements.
In addition, in Section \ref{sec:complexity} we show that, under the same conditions on $B$ and additional realistic assumptions, we expect to have a number of candidate basis elements of the same order of magnitude of the true number of basis elements, $K^c = O(K)$, as the network size $N$ grows, given that the maximum weight is stationary, $\max(Y)=O(1)$.


\subsubsection{Sparse Poisson deconvolution}

Given the candidate set of basis elements encoded in $B^c$, we develop an EM algorithm to estimate the corresponding coefficients $\mu_{1:K^c}$. Under certain conditions on the true  basis matrix $B$, this algorithm consistently estimates the coefficients by zeroing out the coefficients associated with the unnecessary basis elements. 

Recall that we have $K^c$ candidate basis elements $P^c_{1:K^c}$. Let's define a set of statistics $T_k=\sum_{ij}P_{k,ij}^c$ for $k=1\dots K^c$, and let's introduce a set of latent $N\times N$ matrices $G_k$ with positive entries, subject to the only constrain that $\sum_k G_k = Y$.
Algorithm \ref{algo2} details the EM iterative procedure that will lead to maximum likelihood estimates for the coefficients $\mu_{1:K^c}$.
 
\begin{algorithm}[ht!]
Initialize $\mu^{(0)}$\\
\textbf{While} {$||\mu^{(t+1)}-\mu^{(t)}||> \epsilon$}\\
\ \ \ \ \ \ \textbf{For} \ {k=1 to K}\\
\ \ \ \ \ \ \ \ \ \ \ \ \   \textbf{E-step:} ${G_{k,ij}^{(t+1)}}=\mu_k^{(t)}{ Y_{ij} P_{k,ij}^c \over {T_k \sum_{m}\mu^{(t)}_m P_{m,ij}^c}}$\\
\ \ \ \ \ \ \ \ \ \ \ \ \  \textbf{M-step:} ${\mu_k^{(t+1)}}=\sum_{i,j}{G_{k,ij}^{(t+1)}}$;\\
\ \ \ \ \ \ t=t+1;
\protect\smallskip
 \caption{Estimate non-negative coefficients, $\mu$, for all candidate elements, $P^c$, in the basis matrix $B^c$.}
 \label{algo2}
\end{algorithm}

This is the Richardson-Lucy algorithm for Poisson deconvolution \citep{RWBB}. 
The truncated Poisson likelihood can also be maximized directly, using a KL divergence argument \citep{InfoCezar} involving   the discrete probability distribution obtained by normalizing the data matrix, $\bar Y$, and a linear combination of discrete probability distributions obtained by normalizing the candidate basis matrices, $\sum_k \omega_k \bar P_k$.


\section{Theory}
\label{sec:theory}

Here we develop some theory for graphlets when the graph $Y$ is generated by a non-expandable collection of basis matrices. The extent to which these results extend to the general case is unclear, as of this writing, however, the results help form some intuition  about how graphlets encode information. See the Appendix for details of the proofs.
We begin by establishing some notation.

\begin{deff}
\label{def01}
Given two adjacency matrices $P_1$ and $P_2$ both binary define $P_1 \subseteq P_2$ to mean $P_2(j,k)=1$ whenever $P_1(j,k)=1$; i.e. the induced graph by $P_1$ is a subgraph of $P_2$.
\end{deff}

\begin{deff}
\label{def02}
Define $P=\bigvee_{i\in I} P_i$ to mean $P(j,k)=1$ whenever any $P_i(j,k)=1$ and zero otherwise.
\end{deff}

\begin{deff}
\label{def03}
Given $\{P_k\}$, we say that $P'$ is an expansion of $P_k$ if both \ref{def01} and \ref{def02} are satisfied.
\end{deff}

\begin{deff}[non expandable basis]
\label{def1}
A collection of $\{P_k\}$ is non-expandable if for any expansion $P'$ of $P_k$ we can find $j$ such that $P'\subset P_j$.
\end{deff}

A consequence of the above definition is that, if a set of basis is non-expandable, then any subset of it is also non-expandable.


\subsection{Identifiability}
\label{sec:identifiability}

The first result asserts that the collection of Poisson rates $\Lambda$ can be uniquely decomposed into a collection of basis matrices $P_{1:K}$, here encoded by the equivalent binary matrix $B$, with weights $W = \hbox{diag}(\mu_{1:K})$.

\begin{thm}\label{thm1}
Let $\Lambda$ be a symmetric $N\times N$ nonnegative matrix which is generated by a non-expandable basis.  Then there exists a unique $N\times K$ matrix B and $K\times K$ matrix $W$ such that $\Lambda=BWB'$.
\end{thm}

This implies that, in the absence of noise, the set of parameters $B,W,K$ are estimable from the data $Y$.
The proof of Theorem \ref{thm1} is constructive by means of Algorithm \ref{algo0}, which is fast but sensitive to noise.
In practice, we propose a more robust algorithm that uses non-expandability to identify a collection of candidate basis elements $B^c$, which provably includes the unique non expandable basis $B$ that generated the graph.

\begin{thm}\label{thm2}
Let $\Lambda=BWB'$, where $W$ is a diagonal matrix with nonnegative entries and $B$ encodes a non-expandable basis, both unknown. If $B^c$ is the output of Algorithm \ref{algo1} with $\Lambda$ as input, then $B \subseteq B^c$.
\end{thm}

The two theorems above may be used to generate random cliques, for instance, by requiring the entries of the matrix $B$ to be IID Bernoulli random variables.


\subsection{Redundancy and complexity}
\label{sec:complexity}

Here, we derive an upper bound for the maximum number of candidate basis elements encoded in $B^c$ in a network with at most $K=c\log_2N$ cliques. 
Networks with this many cliques include networks generated with most exchangeable graph models \citep{gold:zhen:fien:airo:2010}, and  are the largest networks with an implicit representation \citep{IRGK}.

\begin{thm}\label{thm4}
Let the elements $B_{ik}$ of the basis matrix for a network $Y$ be IID Bernoulli random variables with parameter $p_N$. Then an asymptotic upper bound for the number of candidate basis elements, denoted $C_{N,p_N},$ identified by Algorithm \ref{algo1} is
\begin{eqnarray}
 C_{N,p_N} &\leq& Q (2^{KH(p_N)}+K) \nonumber \\
\label{eqnCN}
           &=   & Q (N^{c_1H(p_N)}+ c\log_2N ),
\end{eqnarray}
where $Q$ is the number of thresholds in Algorithm \ref{algo1}.
\end{thm}

While Theorem \ref{thm4} applies in general, a notion of redundancy of the candidate basis set is well defined only for networks generated by a non expandable basis, in which $B^c \subseteq B$. In this case, we define redundancy as the ratio $R_N\equiv C_{N,p_N}/K$. 
Theorem \ref{thm2} states that number of candidate basis elements is never smaller than the true number of basis $K$. Thus Theorem \ref{thm4} leads to the following upper bound on redundancy 
\begin{equation}
R_{N,p_N} \leq Q \left( 1+{N^{c_1H(p_N)} \over c\log_2N} \right).
\end{equation}

In a realistic scenario we may have  that 
 $p_N = O(1/\log_2 N)$, 
 the complexity of the candidate basis set is $O(Q\log_2 N)$,
 and the implied redundancy is at most $O(Q)$.
Alternatively, if $p_N = O(1/K)$, the implied redundancy is at most $O(QK)$.
These are regimes we often encounter in practice. These limiting behaviors of $p_N$ arise whenever the nodes in a network can be assumed to have a capacity that is essentially independent of the size of the network, e.g., individuals participate in the activities of a constant number of groups over time.

These calculations are only suggestive for networks generated by an expandable basis.
However, non expandability tends to be a reasonable approximation in many situations, including whenever exchangeable graph models are good fit for the network.
In the analysis of messaging patterns on Facebook in Section \ref{sec:fb}, for instance, we empirically observe $C_{N,p} \approx 2\hat K$.


\subsection{Accuracy with less than $K$ basis elements}
\label{appE}

The main estimation Algorithm \ref{algo2} recovers the correct number of basis elements $K$ and the corresponding coefficients $\mu$ and basis matrix $B$, whenever the true weighted network $Y$ is generated from a non expandable basis matrix.
Here we quantify the expected loss in reconstruction accuracy if we were to use $\tilde K < K$ basis elements to reconstruct $Y$.
To this end we introduce a norm for a network $Y$, and related metrics.

\begin{deff}[$\tau$-norm]
\label{def2}
Let $Y\sim Poisson(BWB')$, 
where $W=diag(\mu_1\dots\mu_K)$.
Define the statistics $a_k \equiv \sum^N_{i=1} B_{ik}$ for $k=1\dots K$.
The $\tau$-norm of $Y$ is defined as $\tau(Y) \equiv |\sum_{k=1}^K \mu_k a_k|$.
\end{deff}

Consider an approximation $\tilde Y = B\tilde{W}B'$ characterized by an index set $E \subset \{1\dots K\}$, which specifies the basis elements to be excluded by setting the corresponding set of coefficients $\mu_{E}$ to zero. 
Its reconstruction error is $\tau(Y-\tilde Y) =
 |\sum_{k\notin E} \mu_k  a_k|$,
and its reconstruction accuracy is $\tau(\tilde Y)/\tau(Y)$.
Thus, given a network matrix $Y$ representable exactly with $K$ basis elements, the best approximation with $\tilde K$ basis elements is obtained by zeroing out the lowest $K-\tilde K$ coefficients $\mu$.

We posit the following theoretical model,
\begin{eqnarray}
 \mu_k           & \sim & Gamma(\alpha+\beta,1) \\
 a_k/N             & \sim & Beta(\alpha,\beta) \\
 \mu_k \cdot a_k/N & \sim & Gamma(\alpha,1),
\end{eqnarray}
for $k=1\dots\tilde K$. This model may be used to compute the expected accuracy of an approximate reconstruction based on $\tilde K<K$ basis elements, since the magnitude of the ordered $\mu_k$ coefficients that are zeroed out are order statistics of a sample of Gamma variates \citep{OSEG}. Given $K$ and $\alpha$, we can compute the expected coefficient magnitudes and the overall expected accuracy $\tau_0$.

\begin{thm}
\label{thm-accuracy}
The theoretical accuracy of the best approximate Graphlet decomposition with $\tilde K$ out of $K$ basis elements is:
\begin{equation}
\tau_0(\tilde K,K, \alpha) = \sum^{\tilde K}_{j=1} {f(j,K,\alpha) \over \alpha K},
\end{equation}
where
\begin{eqnarray}
 f(j,K,\alpha)  &=& {K \choose j} \sum^{j-1}_{q=0} (-1)^q {j-1 \choose q} {f(1,K-j+q+1,\alpha)\over K-j+q+1} \\
 f(1,K,\alpha) &=& {K \over \Gamma(\alpha)} \sum^{(\alpha-1)(K-1)}_{m=0} c_m(\alpha,K-1){\Gamma(\alpha+m)\over K^{\alpha+m}},
\end{eqnarray}
in which the coefficients $c_m$ are defined by the recursion
\begin{equation}
 c_m(\alpha,q)=\sum^{i=\alpha-1}_{i=0} {1\over i!}c_{m-i}(\alpha,q-1)
\end{equation}
 with boundary conditions
 $c_{m}(\alpha,1)={1\over i!}$
 for $i=1\dots \alpha$.
\end{thm}

Figure \ref{figa} illustrates this result on simulated networks.
The solid (red) line is the theoretical accuracy computed for $K=30$ and $\alpha=1$, the relevant parameters used to to simulate the sample of 100 weighted networks. The (blue) boxplots summarize the empirical accuracy at a number of distinct values of $\tilde{K}/K$.

\begin{figure}[t!]
\centering
\includegraphics[width=0.85\textwidth]{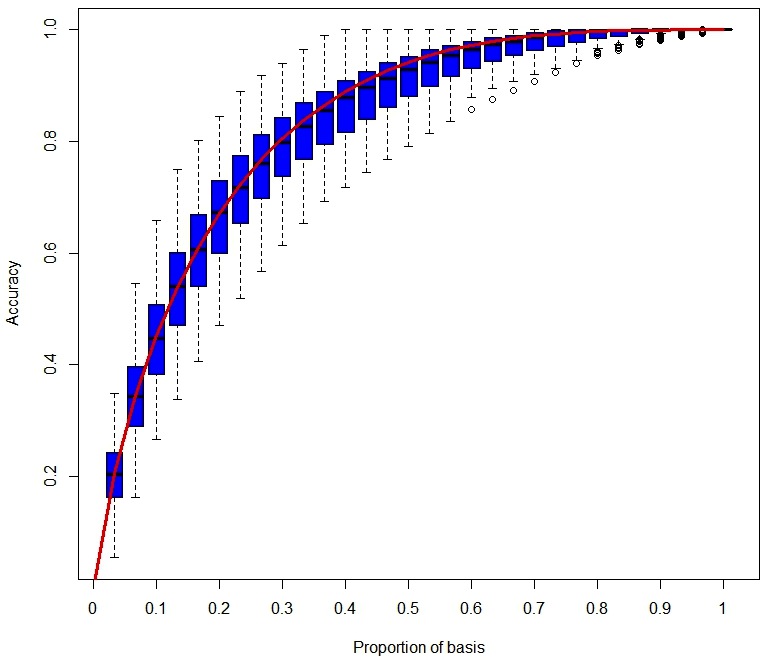}
  \caption{Theoretical and empirical accuracy for different fractions of basis elements $\tilde K / K$ with $\alpha=0.1$. The ratio $\tilde K / K$ also provides a measure of sparsity.}\label{figa}
\end{figure}


\section{Results}

We evaluate graphlets on real and simulated data.
Simulation results on weighted networks in Section \ref{sec:results_estimation} show that the estimation problem is well-posed and both the binary matrix $B$ and the coefficients $\mu$ are estimable without bias, while results on binary networks in Section \ref{sec31} enable the analysis of the comparative performance with existing methods for binary networks.
We also develop an analysis of messaging patterns on Facebook for a number of US colleges and an analysis of historical crime data in Sections  \ref{sec:fb}--\ref{reshis}.

Overall, these results suggest that graphlets encode a new quantification of social information; we provide a concrete illustration of this idea in Section \ref{sec:social_info}.


\subsection{Parameter estimation}
\label{sec:results_estimation}

Here, we evaluate whether the parameters underlying graphlets---the entries of binary matrix $B$, its size $K$, and the non-negative coefficients $\mu$---are estimable from 1 weighted network sample $Y$.

We simulated $M=100$ networks from the model, each with $50$ nodes, using random values for the parameters $\mu,B$ and $K$.
For the $i$-th network, 
 the number of basis elements $K_i$ was sampled from a Poisson distribution with rate $\lambda=30$. 
 The coefficients $\mu_{ji}$ were sampled from a Gamma distribution with parameters $\alpha=1$ and $\beta=10$.
 The entries of the basis matrix $B$ were sampled from a Bernoulli distribution with probability of success $p=0.04$.
The index $i=1\dots M$ runs over the networks and the index $j=1\dots K_i$ runs over the basis elements for the $i$-th network.

Algorithms \ref{algo1}--\ref{algo2} were used to estimate the parameters $B,\mu$ and $K$, for each synthetic network.
The estimation error of $\hat K$ is $|\hat K - K|$.
Since the estimated matrix $\hat B$ is unique only up to a permutation of its columns, the Procrustes transform \citep{PROC} was used to identify the optimal rotation of the matrix $\hat B$ to be compared to the true matrix $B$.
The estimation error of $\hat B$ is  computed as the normalized $L_2$ distance between $B$ and $\hat{B}$ after Procrustes.
The estimation error of $\hat \mu$ is computed as the normalized $L_2$ distance between $\mu$ and $\hat\mu$.
The average error in reconstructing the simulated networks is computed as the $L_1$ distance between $Y$ and $\hat Y$ normalized to the total number of counts in the network, averaged over 100 simulations, denoted $|\tilde Y|_1$. This metric gives more weight to larger cliques, since the clique size enters as a quadratic term.
We also provide the average error in reconstructing the simulated networks based on the $\tau$-norm, which is less sensitive to incorrectly estimating larger cliques.
In addition, we compute the average error in reconstructing whether or not edges have positive weights, rather than their magnitude, denoted error in $\mathbf{1}(\tilde Y>0)$.
 
Table \ref{table1} summarizes the simulation results for weighted network reconstructions obtained by setting a target accuracy $\tau_0$, ranging from 0.20 to 1.00---no error.
The last row of Table \ref{table1} supports the claim that the estimation problem is well-posed; the estimation algorithm stably recovers the true parameter values.
The only sources of non-zero error are the estimation of the number of basis elements and the estimation of the basis elements themselves, as encoded by the matrix $B$. However, these errors are negligible and do not seem to impact the estimation of the coefficients $\mu$, nor to propagate to the network reconstruction independently of the metric we use to quantify it.

The average bias in estimating the optimal number of basis elements $K$ is 0.07, with a standard deviation of 0.256.
Column $\tilde K/ \hat K$ provides the fraction of the estimated optimal number of basis elements $\hat K$ that is necessary to achieve the desired target accuracy, in terms of the $\tau$-norm, using Theorem \ref{thm-accuracy}. Column two provides the empirical $\tau(\tilde Y)$ error, defined as one minus the accuracy, corresponding to the reported fraction of basis elements.
The empirical accuracy matches the expected accuracy given by Theorem \ref{thm-accuracy} well.
For instance, the Theorem \ref{thm-accuracy} implies that we need 45\% of the basis elements in order to achieve an accuracy of 0.85, and the empirical accuracy is slightly in excess of 0.85, on average.
Figure \ref{figa} shows similar results, for a larger set of 29 values for the target accuracy $\tau_0$.

The first two columns of Table \ref{table1} report the average $L_1$ error and the average $\tau$-based error, properly normalized to have range in zero to one. Cliques in the simulated networks range between 2 and 5 nodes, for most networks. The differences in the reconstruction error are due to the weight of these larger cliques. For instance, the empirical $L_1$ accuracy obtained by using 45\% of the basis elements is around 0.65, as opposed to a $\tau$ accuracy of around 0.85, on average.

\begin{table*}[t!]
\caption{Performance of the estimation algorithm for a target accuracy $\tau_0$, on synthetic data.}
\label{table1}
\begin{center}
\begin{tabular}{c c c c c c c c}
\hline
\aroundspace
 $|\tilde Y|_1$ error & $\tau(\tilde Y)$ error & error in $\mathbf{1}(\tilde Y>0)$ & error in $B$ & error in $\mu$ & $\tilde K/ \hat K$ & $\tau_0$ \\
\hline
\abovespace
 $0.97\pm0.012$ & $0.65\pm0.074$ & $0.76\pm0.062$ & $50.75\pm10.893$ & $0.55\pm0.082$ & $0.08\pm0.035$ & $0.20$ \\
 $0.92\pm0.016$ & $0.46\pm0.026$ & $0.61\pm0.053$ & $42.40\pm~\,8.134$ & $0.32\pm0.041$ & $0.16\pm0.039$ & $0.50$ \\
 $0.79\pm0.036$ & $0.23\pm0.012$ & $0.40\pm0.057$ & $29.14\pm~\,5.536$ & $0.10\pm0.015$ & $0.33\pm0.050$ & $0.75$ \\
 $0.66\pm0.054$ & $0.14\pm0.008$ & $0.30\pm0.053$ & $21.72\pm~\,4.441$ & $0.04\pm0.009$ & $0.45\pm0.060$ & $0.85$ \\
 $0.56\pm0.063$ & $0.09\pm0.006$ & $0.24\pm0.048$ & $17.04\pm~\,3.757$ & $0.02\pm0.004$ & $0.54\pm0.062$ & $0.90$ \\
 $0.41\pm0.064$ & $0.04\pm0.004$ & $0.17\pm0.041$ & $11.36\pm~\,2.980$ & $0.00\pm0.002$ & $0.67\pm0.061$ & $0.95$ \\
 $0.19\pm0.051$ & $0.01\pm0.002$ & $0.07\pm0.024$ & $~\,4.46\pm~\,1.646$ & $0.00\pm0.000$ & $0.85\pm0.043$ & $0.99$ \\
\belowspace
 $0.00\pm0.000$ & $0.00\pm0.000$ & $0.00\pm0.000$ & $~\,0.67\pm~\,0.618$ & $0.00\pm0.000$ & $1.01\pm0.033$ & $1.00$ \\
\hline
\end{tabular}
\end{center}
\end{table*}

The accuracy in reconstructing the presence or absence of edges in the network, rather than the magnitude of the corresponding weights, is in between the $L_1$ accuracy and the $\tau$ accuracy, for any fraction $\tilde K/ \hat K$.


\subsection{Link reconstruction}
\label{sec31}

Graphlet is a method for decomposing a weighted network. 
However, here we evaluate how graphlet fares in terms of the computational complexity versus accuracy trade-off when compared to the performance current methods for binary networks can achieve.

For this purpose, we generated 100 synthetic networks from the model, each with $100$ nodes, choosing random values for the parameters as we did in Section \ref{sec:results_estimation}.
 The number of basis was sampled from a Poisson  with rate $\lambda=12$. 
 The coefficients were sampled from a Gamma  with parameters $\alpha=2$ and $\beta=10$.
 The entries of the basis matrix were sampled from a Bernoulli  with probability of success $p=0.1$.
We then obtained the corresponding binary networks by resetting positive edge weights to one, $\mathbf{1}(Y>0)$.

\begin{table}[t!]
\caption{Link reconstruction accuracy and runtime for a number of competing methods, on simulated data. The graphlet decomposition is denoted Glet, with the fraction of basis elements used quoted  in brackets.}
\label{table2}
\begin{center}
\begin{tabular}{ l c r }
\hline
\aroundspace
Method & Runtime (sec) & Accuracy \\
\hline
\abovespace
{Glet (25\%)}  & $ 0.0636$ & $ 92.7 \pm 0.30$ \\
{Glet (50\%)}  & $ 0.0636$ & $ 94.7 \pm 0.15$ \\
{Glet (75\%)}  & $ 0.0636$ & $ 97.0 \pm 0.08$ \\
{Glet (90\%)}  & $ 0.0636$ & $ 98.9 \pm 0.05$ \\
\belowspace
{Glet (100\%)} & $ 0.0636$ & $100.0 \pm 0.00$ \\
\hline
\abovespace
{ERG}          & $ 0.0127$ & $ 86.0 \pm 1.20$ \\
{DDS}          & $11.1156$ & $ 89.0 \pm 0.85$ \\
{LSCM}         & $ 6.3491$ & $ 90.0 \pm 0.70$ \\
\belowspace
{MMSB}         & $ 2.8555$ & $ 93.0 \pm 0.40$ \\
\hline
\end{tabular}
\end{center}
\end{table}

We compute comparative performance for a few popular models of networks that can be used for link prediction.
These include 
 the Erd\"os-R\'enyi-Gilbert random graph model \citep{Erdo:Reny:1959,Gilb:1959}, denoted ERG,
 a variant of the degree distribution model fitted with importance sampling \citep{blit:diac:2010}, denoted DDS,
 the latent space cluster model \citep{Hand:Raft:Tant:2007}, denoted LSCM,
 and the stochastic blockmodel with mixed membership \citep{Airo:Blei:Fien:Xing:2008}, denoted MMSB.
A comparison with singular value decomposition in terms binary link prediction is problematic, since using a few singular vectors would not lead to binary edge weights.

Table \ref{table2} summarizes the comparative performance results.
The accuracy of graphlet is reported for different fractions of the estimated optimal number of basis elements $\hat K$, ranging from 25\% to 100\%---no error.
The accuracy is consistently high, while the running time is a only a fraction of a second, independent of the desired reconstruction accuracy. This is because the complexity of computing the {\em entire} graphlet decomposition is linear in the number of edges present in the network. Thus specifying a smaller fraction of the number of basis elements is a choice based on storage considerations, rather than on runtime.
The comparative performance results suggest that graphlet decomposition is very competitive in predicting links.

This binary link prediction simulation study is one way to compare graphlet to interesting existing methods, which have been developed for binary graphs.



\subsection{Analysis of messaging on Facebook}
\label{sec:fb}

Here we illustrate graphlet with an application to messaging patterns on Facebook.
We analyzed the number of public wall-post on Facebook, over a three month period, among students of a number of US colleges. While our data is new, the US colleges we selected have been previously analyzed \citep{trau:kels:much:port:2011}.

\begin{table*}
\caption{$\tau(\tilde Y)$ error for different fractions of the estimated optimal number of basis elements $\hat K$.}\label{table3}
\begin{center}
\begin{tabular}{l r r r r r r r r r r r r}
\hline
\aroundspace
college & nodes & edges & $\hat K$ & sec & 10\% & 25\% & 50\% & 75\% & 90\% & 100\% \\ \hline
\abovespace
 {American} &  {6386}& {435323} &  {11426} & {151}& {13.5}& {6.5} &  {1.80} & {.70}& {.30}& {0}\\
 {Amherst} &  {2235}& {181907} &  {10151} & {124}& {6.5}& {2.3} &  {.84} & {.33}& {.14}& {0}\\
 {Bowdoin} &  {2252}& {168773} &  {9299} & {113}& {9.0}& {3.2} &  {1.17} & {.41}& {.17}& {0}\\
 {Brandeis} &  {3898}& {275133} &  {10340} & {116}& {6.8}& {2.9} &  {1.21} & {.53}& {.25}& {0}\\
 {Bucknell} &  {3826}& {317727} &  {13397} & {193}& {7.8}& {2.9} &  {1.15} & {.45}& {.20}& {0}\\
 {Caltech} &  {769}& {33311} &  {3735} & {51}& {5.7}& {1.8} &  {.65} & {.27}& {.11}& {0}\\
 {CMU} &  {6637}& {499933} &  {11828} & {169}& {14.9}& {5.2} &  {1.89} & {.73}& {.34}& {0}\\
 {Colgate} &  {3482}& {310085} &  {12564} & {151}& {8.0}& {3.3} &  {1.26} & {.45}& {.19}& {0}\\
 {Hamilton} &  {2314}& {192787} &  {11666} & {200}& {6.9}& {2.5} &  {.84} & {.33}& {.15}& {0}\\
{Haverford76} &  {1446}& {119177} &  {9021} & {128}& {4.8}& {2.2} &  {.74} & {.25}& {.10}& {0}\\
 {Howard} &  {4047}& {409699} &  {12773} & {170}& {8.6}& {3.7} &  {1.55} & {.60}& {.28}& {0}\\
 {Johns Hopkins} &  {5180}& {373171} &  {11674} & {150}& {10.8}& {3.7} &  {1.40} & {.58}& {.29}& {0}\\
 {Lehigh} &  {5075}& {396693} &  {14076} & {206}& {9.5}& {3.2} &  {1.14} & {.49}& {.23}& {0}\\
 {Michigan} &  {3748}& {163805} &  {5561} & {54}& {11.4}& {4.6} &  {1.92} & {.76}& {.35}& {0}\\
 {Middlebury} &  {3075}& {249219} &  {9971} & {109}& {9.7}& {3.5} &  {1.35} & {.49}& {.22}& {0}\\
 {MIT} &  {6440}& {502503} &  {13145} & {191}& {11.5}& {4.7} &  {1.68} & {.65}& {.30}& {0}\\
 {Oberlin} &  {2920}& {179823} &  {7862} & {84}& {9.8}& {3.8} &  {1.48} & {.57}& {.26}& {0}\\
 {Reed} &  {962}& {37623} &  {3911} & {46}& {6.0}& {2.3} &  {.99} & {.40}& {.17}& {0}\\
 {Rice} &  {4087}& {369655} &  {12848} & {155}& {8.7}& {3.2} &  {1.25} & {.51}& {.22}& {0}\\
 {Rochester} &  {4563}& {322807} &  {10824} & {124}& {11.0}& {3.7} &  {1.43} & {.56}& {.26}& {0}\\
 {Santa} &  {3578}& {303493} &  {11203} & {127}& {10.3}& {3.5} &  {1.30} & {.51}& {.24}& {0}\\
 {Simmons} &  {1518}& {65975} &  {5517} & {60}& {6.4}& {2.5} &  {1.04} & {.44}& {.19}& {0}\\
 {Smith} &  {2970}& {194265} &  {8591} & {102}& {5.5}& {2.5} &  {1.15} & {.49}& {.23}& {0}\\
 {Swarthmore} &  {1659}& {122099} &  {7856} & {96}& {6.3}& {2.6} &  {1.06} & {.44}& {.20}& {0}\\
 {Trinty} &  {2613}& {223991} &  {10832} & {131}& {8.6}& {3.0} &  {1.07} & {.40}& {.18}& {0}\\
 {Tufts} &  {6682}& {499455} &  {14641} & {212}& {13.9}& {4.8} &  {1.68} & {.65}& {.30}& {0}\\
 {UC Berkeley} &  {6833}& {310663} &  {7715} & {105}& {16.3}& {6.2} &  {2.28} & {.90}& {.42}& {0}\\
 {U Chicago} &  {6591}& {416205} &  {12326} & {176}& {14.2}& {4.7} &  {1.74} & {.66}& {.31}& {0}\\
 {Vassar} &  {3068}& {238321} &  {11344} & {134}& {9.3}& {3.2} &  {1.18} & {.47}& {.21}& {0}\\
 {Vermont} &  {7324}& {382441} &  {10030} & {145}& {17.5}& {5.4} &  {2.05} & {.82}& {.36}& {0}\\
 {USFCA} &  {2682}& {130503} &  {6735} & {67}& {10.5}& {3.8} &  {1.50} & {.60}& {.26}& {0}\\
 {Wake Forest} &  {5372}& {558381} &  {15580} & {211}& {11.5}& {4.2} &  {1.46} & {.58}& {.25}& {0}\\
 {Wellesley} &  {2970}& {189797} &  {9768} & {107}& {8.9}& {3.3} &  {1.24} & {.48}& {.22}& {0}\\
\belowspace
 {Wesleyan} &  {3593}& {276069} &  {10506} & {118}& {9.9}& {3.6} &  {1.40} & {.54}& {.25}& {0}\\
\hline 
\end{tabular}
\end{center}
\end{table*}

Table \ref{table3} provides a summary of the weighted network data and of the results of the graphlet decomposition.
Salient statistics for each college include the number of nodes and edges.
The table reports the number of estimated basis elements $\hat K$ for each network and the runtime, in seconds.
The $\tau(\tilde Y)$ error incurred by using a graphlet  decomposition is reported for different fractions of the estimated optimal number of basis elements $\hat K$, ranging from 10\% to 100\%---no error.

Overall, these results suggest that the compression of wall-posts achievable on collegiate network is substantial.
A graphlet decomposition with about 10\% of the optimal number of basis elements already leads to a reconstruction error of 10\% or less, with a few exceptions. Using 25\% of the basis elements further reduces the reconstruction error to below 5\%.


\subsection{Analysis of historical crime data}
\label{reshis}

More recently, we have used the graphlet decomposition to analyze crime associations records from the 19th century.
These records are available at the South Carolina Department of Archives and History (Columbia, SC), Record Group 44, Series L 44158, South Carolina, Court of General Sessions (Union County) Indictments 1800-1913. 
Nodes in the network are 6221 residents of Union County, during the years 1850 to 1880.
Edge weights indicate the number of crimes involving pairs of residents.
Details on the structure of this data set and  its historical relevance may be found in \cite{EFPM,EFPK}.
\begin{figure}
 \centering
  \includegraphics[width=0.75\textwidth]{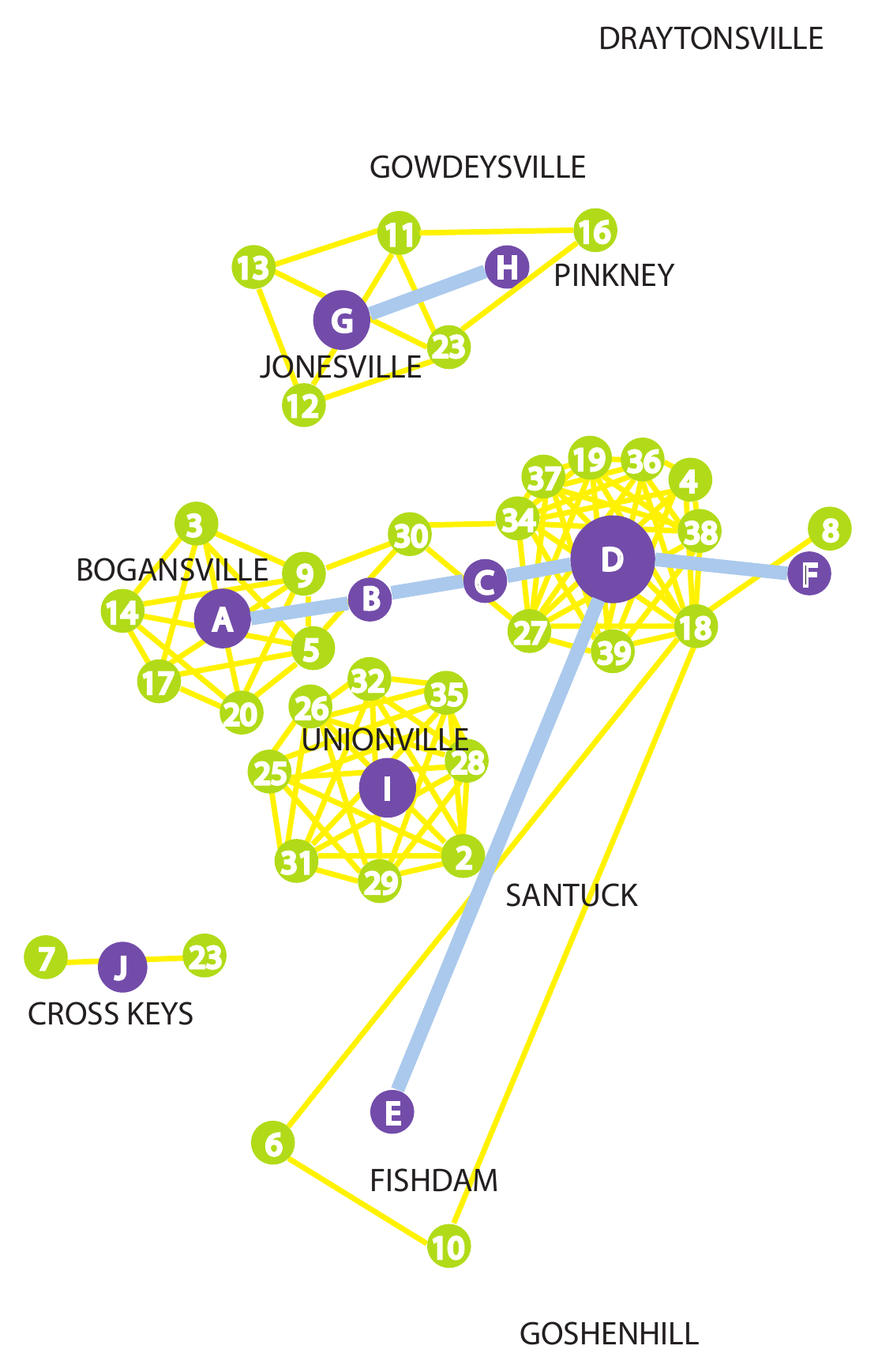}
  \caption{A two-step analysis of the Union County criminal association network using graphlets reveals tight pockets of criminality (the numbered nodes) and local criminal communities (the lettered nodes).}\label{fig8}
\end{figure}

A first step in the analysis was to explore the degree to which there is an optimal level of granularity, for the observed associations, at which a non-trivial criminal structure emerges.
We pursued this analysis by fitting graphlets to the raw criminal association network raised to different powers, $Y^k$ for $k=1\dots 5$.
We found that $Y^2$ contained a substantial degree of criminal social structure. 
At this resolution, $k=2$, two residents are associated with crimes in which they are either directly involved, or in which  they are involved indirectly  through their direct criminal associates.
The results suggest that there are several tight pockets of criminality in Union County.

We were interested in developing an interpretation for these pockets of criminality.
One compelling hypothesis is that each pocket captures a gang or a clan. If this holds, we might expect the variation in the gangs' relations to be explained by the spatial organization of the townships they operate in, to a large degree.

To test this hypothesis, we built a meta network of criminal associations, in which each gang identified in the first step of the analysis is represented by a meta node. 
Two or more gangs were connected in the meta network whenever a criminal record involved at least one of their members.
A graphlet decomposition of the meta network identified pockets criminal activity at the gang level.
Plotting these gang level criminal associations reveals the townships' spatial organization, to a large degree.
Figure \ref{fig8} shows the pockets of criminality identified by the graphlet analysis of the original criminal associations, $Y^2$, as green nodes associated with numbers, as well as the pockets of criminality identified by the graphlet analysis of the meta network, as  purple nodes associated with letters.

In summary, we were abel to successfully use graphlets  as an exploratory tool to capture and distill an interesting, multi-scale organization among criminals in 19th century Union County, from historical records.

\section{Discussion}

Taken together, our results suggest that the graphlet decomposition implies a new notion of social information, quantified in terms of multi-scale social structure. We  explored this idea with a simulation study. 
 

\subsection{A new notion of ``social information''}
\label{sec:social_info}

Graphlet quantifies social information in terms of community structure (i.e., maximal cliques) at multiple scales and possibly overlapping.
To illustrate this notion of social information, we simulated 200 networks: 100 Erd\"os-R\'enyi-Gilbert random graphs with Poisson weights and 100 networks from the model described in Section \ref{secmod}.
While, arguably, the Poisson random graphs do not contain any such information,  the data generating process underlying graphlets was devised  to translate such information into edge weights.

As a baseline for comparison, we consider the singular value decomposition \citep{sear:2006}, 
applied to the symmetric adjacency matrices with integer entries encoding the weighted networks.
The SVD summarizes edge weights  using orthogonal eigenvectors $v_i$ as basis elements and the associated eigenvalues $\lambda_i$ as weights, $Y_{N\times N} = \sum_{i=1}^N \lambda_i \, v_i \, v_i'$.
However, SVD basis elements are not interpretable in terms of community structure, thus SVD should not be able to capture the notion of social information we are interested in quantifying.

\begin{figure*}[b!]
\centering
\includegraphics[width=.495\textwidth]{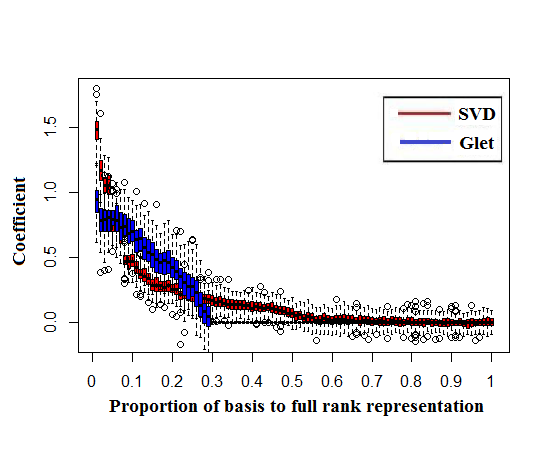}
\includegraphics[width=.495\textwidth]{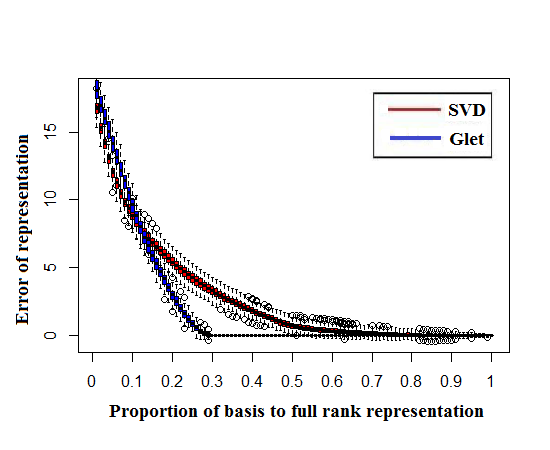}
\includegraphics[width=.495\textwidth]{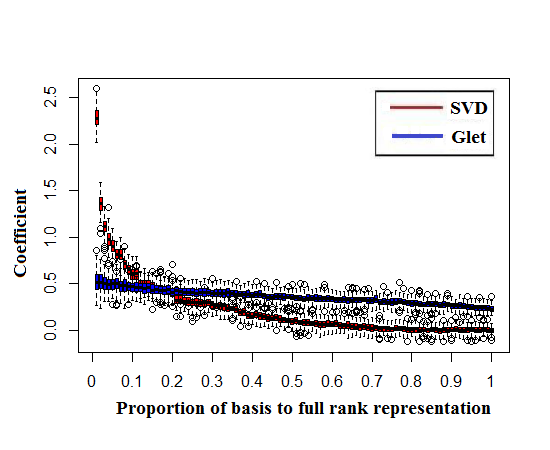}
\includegraphics[width=.495\textwidth]{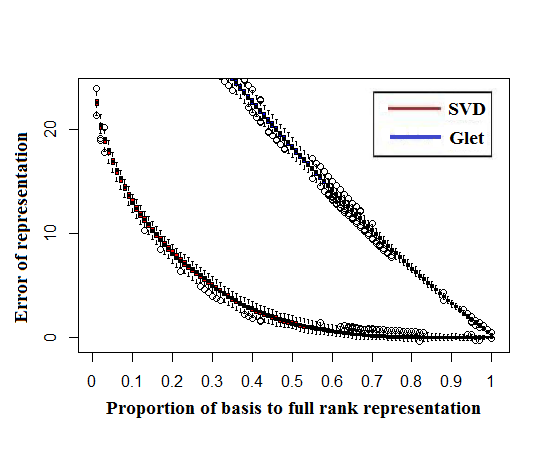}
\caption{Comparison between  graphlet and  SVD decompositions. Top panels report results for on the networks simulated from the model underlying graphlets. Bottom panels  report results on the Poisson random graphs.}\label{fig:BSDstructure1}
\end{figure*}

We applied graphlets and SVD to the two sets of networks we simulated.
Figure \ref{fig:BSDstructure1} provides an overview of the results.
Panels in the top row  report results for on the networks simulated from the model underlying graphlets. 
Panels in the bottom row report results on the Poisson random graphs. 
In each row,
 the left panel shows  the box plots of the coefficients associated with each of the basis elements, for graphlets in blue and for SVD in red.
 The right panel shows  the box plots of the cumulative representation error as a function of the number of basis elements utilized. 
Graphlets coefficients decay more slowly than SVD coefficients on Poisson graphs (bottom left).
Because of this, the error in reconstructing Poisson graphs achievable with graphlets is consistently worse then the error achievable with SVD (bottom right).
In contrast, graphlets coefficients decay rapidly to zero on networks with social structure, much sharply then the SVD coefficients (top left).
Thus, the reconstruction error achievable with graphlets is consistently better then the error achievable with SVD (top right).

These results support our claim that graphlets is able to distill and quantify a notion of social information in terms of social structure.


\subsection{Concluding remarks}

The graphlet decomposition of a weighted network is an exploratory tool,
 which is  based on a simple but compelling statistical model, 
 it is amenable to theoretical analysis,
 and it scales to large weighted networks.

SVD is the most parsimonious orthogonal decomposition of a data matrix $Y$ \citep{Hast:Tibs:Frie:2001}.
Graphlets abandon the orthogonality constraint to attain interpretability of the basis matrix in terms of social information.
In the presence of social structure, graphlet is a better summary of the variability in the observed connectivity. 
SVD always achieves a lower reconstruction error, however, when a few basis elements are considered (see Figure \ref{fig:BSDstructure1}) suggesting that orthogonality induces a more useful set of constraints whenever a very low dimensional representation is of interest.

Graphlets provide a parsimonious sumary of complex social structure, in practice. 
 The simulation studies in Sections \ref{sec:results_estimation} and \ref{sec:social_info} suggest that graphlets are leveraging the social structure underlying messaging patterns  on Facebook  to deliver the high compression/high accuracy ratios we empirically observe in Table \ref{table3}.  
 
The information graphlets utilize for inference comes exclusively from positive weights in the network.
This feature is the source of desirable properties, however, because of this, graphlets are also sensitive to missing data; zero edge weights cannot be imputed.
More research is needed to address this issue properly, while maintaining the properties outlined in Section \ref{sec:theory}.
Empirical solutions include 
 considering powers of the original network, as in Section \ref{reshis},
 or cumulating weights over a long observation window, as in Section \ref{sec:fb}.
Intuitively, cumulating messages over long time windows will reveal a large set of basis elements.
Algorithm \ref{algo2} can then be used to project message counts over a short observation window onto the larger set of basis elements to estimate which elements were active in the short window and how much. 
Recent empirical studies suggest three months as a useful cumulation period \citep{VOCI,Koss:Watt:2006}.

\subsection*{Acknowledgments}

The authors wish to thank Prakash Balachandran, Jonathan Chang and Cameron Marlow.
This work was partially supported by the National Science Foundation under grants DMS-1106980, IIS-1017967, and CAREER award IIS-1149662, and by the Army Research Office under grant 58153-MA-MUR, all to Harvard University. EMA is an Alfred P. Sloan Research Fellow.


\bibliographystyle{plainnat}


\appendix
\section{Proofs of theorems and lemmas}

In the following, $\Lambda \circ C$ denotes the element-wise (Hadamard) product of two matrices $\Lambda$ and $C$.


\subsection{Proof of Theorem \ref{thm1}}\label{appthm1}

\begin{proof}
We will show $B$ can be recovered from $\Lambda$ using the deterministic algorithm \ref{algo0}.

\begin{algorithm}[!h]\label{algo0}
Set $\Lambda'=\Lambda$\\
\ \ \textbf{While}{(There is at least one non-zero element in $\Lambda'$)}\\
\ \ \ \ \ \ \ \   $\mu_c=\min_{i,j}(\Lambda \circ C)(i,j)$\\
\ \ \ \ \ \ \ \   ${ij}_c$=$ \arg\min_{ij} (\Lambda \circ C)(i,j)$\\
\ \ \ \ \ \ \ \   C=largest clique in $\Lambda'$ that includes the edge ${ij}_c$\\
\ \ \ \ \ \ \ \   $\Lambda'$=$\Lambda'-\mu_c C$\\
\ \ \ \ \ \ \ \   Add $(C,\mu_c)$ to $(B,W)$ respectively.\\
\protect\smallskip
\caption{Algorithm in theorem \ref{thm1} for uniqueness of non expandable basis}
\end{algorithm}

In each step, due to non-expandability, the maximal clique of $\Lambda'$, $C$, will not be a subset of any other clique and will contain an edge which only belongs to $C$, using Lemma \ref{lem0} below. The weight of this edge in $\Lambda'$ will be $\mu_c$ the coefficient of clique $C$ in the composition. Also, edges of clique $C$ in the network will not have weights less than $\mu_c$ because all the coefficients in composition($\mu$s) are positive. So, in each step, the weight corresponding to the edge(s) unique to $C$ will be correctly detected and subtracted from the network.

Hence, using algorithm \ref{algo0} we can find basis $B$ and coefficients $W$ of matrix in $\Lambda'=BWB^T$. Since the algorithm is deterministic given $\Lambda$, it will provide us only one set of unique NEB.
\end{proof}


\begin{lem}\label{lem0}
If B generates a non-expandable basis set $S$, any clique $C \in S$ will contain an edge that does not belong to any other clique in $S$.
\end{lem}

\begin{proof}
Proof by contradiction: If every edge in $C$ overlaps with another clique in $S$, then $C$ can be decomposed to other cliques, which contradicts non-expandability of $S$.
\end{proof}


\subsection{Proof of Theorem \ref{thm2}}\label{appthm2}

\begin{proof}
Recall that $P_{k}=b_{.k}b_{.k}'$ and $\Lambda=\sum_l \mu_l P_l$. For any $k \in \{1,2,..,K\}$, setting a threshold on the edge weights of the network $\Lambda$, such that edges with weights above the threshold are selected, obtains a binary network. 
Define $G_k$ to be those indices $\ell$ for which $B_{\ell}$ covers $B_{k}.$
If the above mentioned threshold is set to $th_k=\sum_{l\in G_k}\mu_l$, then a binary network $\Lambda_{th=th_k}$ will be revealed. Below we will prove that $P_{k}$ is a maximal clique for $\Lambda_{th=th_k}$ using Lemma \ref{lem2}.

To show that $P_{k}$ is a maximal cliques of $\Lambda_{th=th_k}$ we need to prove the following two statements:

(1) $P_{k} \subseteq \Lambda_{th=th_k}$ which means clique $P_k$ should be included in $\Lambda_{th=th_k}$. 
This statement is true because the weights of the edges in $\Lambda$ that correspond to $P_{k}$, are greater than $th_k$ since $\mu$s are positive.

(2) $P_{k}$ can not be part of a larger clique in $\Lambda_{th=th_k}$. 
By contradiction, assume such a larger clique exists and name it $P'$ then $P_{k} \subset P' \subseteq \Lambda_{th=th_k}$. Considering that the basis is non-expandable, given lemma (\ref{lem2})  and $P' \subseteq \Lambda_{th=th_k} \subseteq \bigvee_{i \in T \backslash G} P_{l_i}$ then we can find $j \in T\backslash G$ that $P' \subseteq P_{j}$. Hence, we have $P_{k} \subset P' \subseteq P_{j}$ which leads to $P_{k}\subset P_{j}$ for a $j \in T \backslash G$. This is clearly a contradiction since such $j$'s are excluded in $G$.

With $P_k$'s being one of the maximal cliques of thresholded binary networks at $th_k$ and considering the fact that the threshold level will be equal to $\sum_{i\in I}\mu_i$ for any $I \in \{G_1,G_2,...,G_K\}$ (because $\min_{ij} \Lambda_{ij} \le \sum_{l\in I}\mu_l \le \max_{ij} \Lambda_{ij}$ ), 
it is clear that $B\subseteq B_c$.
This ends the proof of theorem (\ref{thm2}).
\end{proof}


\begin{lem}\label{lem2}
$\Lambda_{th=th_k}\subseteq \bigvee_{i \in T \backslash G_k} P_{i}$, for $T=1\dots K$.
\end{lem}

\begin{proof}
We will prove this by contradiction. Assume the above claim is not true. In that case, an edge in $\Lambda_{th=th_k}$ can be found which belongs to complementary set: $\bigvee_{i \in G_k} P_{i}\backslash \bigvee_{i \in T \backslash G_k} P_{i}$. However, the weight for this edge is more than $th_k=\sum_{l\in G_k}\mu_l$ which is not possible for edges included in $\bigvee_{i \in G_k} P_{i}\backslash\bigvee_{i \in  T\backslash G_k} P_{i}$ because the maximum weight for edges obtained from this group of cliques is at most $th_k=\sum_{l\in G_k}\mu_l$. End of proof of the lemma (\ref{lem2}).
\end{proof}


\subsection{Proof of Theorem \ref{thm4}}\label{appthm3}

\begin{proof}
Using the definition of Graphlet, $P_1,..,P_K$ are cliques generated from matrix $B$. Considering that $Y=\sum_l \mu_lP_l$, any clique $P'$ detected at different thresholds will be in the form of $P'=\bigwedge_{l \in I} P_l$ for an $I \subseteq \{1,2,..,K\}$. Each clique will be detected at most $Q$ times. Then the total number of detected cliques will not be more than $Q$ times $2^K$, the total number of possible $I$s. However, based on the distribution of elements of matrix $B$ we know the typical set for $I$s have size of $2^{KH(p_N)}$, asymptotically for large $n$, where $H(.)$ is entropy function. This follows the fact that the bit strings for nodes will be in a typical set where each bit-string has $Kp_N$ bits of one. Hence, the number of nonempty common intersections among a nontypical set of cliques will converge to zero and we will have $C_{N,p}=Q2^{KH(p_N)}+QK$.
\end{proof}


\end{document}